\documentclass[12pt]{article}

\usepackage{amsmath,amssymb,amsfonts,amscd, amsthm}
\usepackage{dsfont}
\usepackage{enumerate}
\usepackage{mathabx}
\usepackage{natbib}

\usepackage{tikz}
\usepackage{caption}
\usetikzlibrary{decorations.pathreplacing}

\usetikzlibrary{cd} 
\usepackage{tkz-graph}
\usepackage{pgfplots}
\usepackage{subcaption}
\usetikzlibrary{calc}
\usepgfplotslibrary{fillbetween}
\usetikzlibrary{patterns, positioning, arrows}
\usetikzlibrary{arrows.meta}
\usepackage{schemabloc}
\usetikzlibrary{decorations.markings}
\pgfplotsset{compat=1.14}
\usepackage{marginnote}
\usepackage{setspace}
\usepackage{xcolor}

\usepackage{color}
\definecolor{DarkGreen}{rgb}{0.2,0.6,0.2}

\definecolor{purple}{rgb}{0.6,0.3,0.8}

\usepackage[colorlinks=true,citecolor=blue]{hyperref}
\usepackage{graphicx}

\theoremstyle{plain}

\newtheorem{theorem}{Theorem}
\newtheorem*{theorem*}{Theorem}
\newtheorem{proposition}{Proposition}
\newtheorem*{proposition*}{Proposition}

\theoremstyle{definition}

\newtheorem{lemma}{Lemma}
\newtheorem*{lemma*}{Lemma}

\newtheorem{definition}{Definition}
\newtheorem*{definition*}{Definition}

\theoremstyle{remark}
\newtheorem{example}{Example}
\newtheorem*{example*}{Example}

\newtheorem*{assumption*}{Assumption}
\newtheorem{fact}{Fact}
\newtheorem*{fact*}{Fact}
\newtheorem{corollary}{Corollary}
\newtheorem*{corollary*}{Corollary}

\newtheorem*{question*}{Question}

\newtheorem{remark}{Remark}

\def\laweq{\buildrel \mathrm{d} \over =}

\renewcommand{\d}{\mathrm{~d}}
\newcommand{\VaR}{\mathrm{VaR}}

\newcommand{\ES}{\mathrm{ES}}

\renewcommand{\P}{\mathbb{P}}
\newcommand{\E}{\mathbb{E}}
\newcommand{\R}{\mathbb{R}}

\renewcommand{\L}{\mathcal{L}}

\usepackage{booktabs}
\begin{document}

\title{Partial comonotonicity and distortion riskmetrics}
\author{Muqiao Huang\thanks%
  {Department of Statistics and Actuarial Science,
  University of Waterloo, Canada.
 \href{mailto:m5huang@uwaterloo.ca}{m5huang@uwaterloo.ca}.}} 
\maketitle
\begin{abstract} We establish a connection between dependence structures and subclasses of distortion riskmetrics under which the latter are additive. A new notion of positive dependence, called partial comonotonicity, is developed, which nests the existing concepts of comonotonicity and single-point concentration. 
For two random variables, being comonotonic with a third one does not imply that they are comonotonic; instead, this defines an instance of partial comonotonicity. 
Any specific instance of partial comonotonicity uniquely characterizes a class of distortion riskmetrics through additivity under this dependence structure.  
An implication of this result is the characterization of the Expected Shortfall using single-point concentration.
\end{abstract}

\section{Introduction}\label{sec:intro}

Positive dependence among risks is an important consideration in risk management (\cite{MFE15}). 
As the extreme form of positive dependence, comonotonicity models perfect comovement between random losses from a financial portfolio.
Additivity for comonotonic losses is a key property in the characterization of many risk measures and functionals, widely used in economics, finance and insurance (\cite{S89, WYP97, K01}). A weaker notion of positive dependence, known as $p$-concentration, models the following scenario: there is some catastrophic event with probability $(1-p)$, under which all assets in a portfolio perform worse than they would have if such an event had not occurred. In a characterization of the Expected Shortfall at level $p$ (\cite{WZ21}), $p$-concentration played a central role. 
The current paper aims to characterize subclasses of distortion risk measures, and more generally, distortion riskmetrics (not requiring the distortion function to be increasing, defined in \cite{WWW20b}),  via their additivity for specific dependence structures. We achieve the characterization by using a new collection of positive dependence structures called partial comonotonicity.

 Notably, $p$-concentration, upper comonotonicity (\cite{C09}), and lower, tail, and interval comonotonicity (\cite{ZD13}), as well as comonotonicity, are particular cases of partial comonotonicity. A broader generalization called weak comonotonicity was introduced in \cite{WZ20}. There, random variables are required to move in the same direction on average in a weighted sense. By choosing a large number of simple weights, partial comonotonicity can be recovered.

In this paper, we will define $K$-concentration for subsets $K\subseteq[0,1]$ and $g$-comonotonicity for left-continuous increasing functions $g$ on $[0,1]$. These are referred to as instances of partial comonotonicity. A dependence structure $D$ may refer to a copula or an instance of partial comonotonicity.
We will see that an instance $D$ of partial comonotonicity corresponds to a collection $\mathcal C_D$ of copulas. That is, a random vector has the dependence structure $D$ (from the class of partial comonotonicity) if and only if its copula is in $\mathcal C_D$.  For a dependence structure $D$, a functional $f$ is called $D$-additive if for any random variables $(X,Y)$ satisfying $D$, $f$ is additive for $(X,Y)$.
When $D$ is chosen as comonotonicity, 
we arrive at the usual comonotonic additivity, the defining property of Choquet integrals (\cite{S86}). The law-invariant case was characterized by \cite{Y87}. In this case, the Choquet integrals are also known as distortion risk measures, studied by \cite{DDGKV02, DVGKTV06, DKLT12}, and their non-monotone version is studied by \cite{WWW20}. On the other hand, $p$-concentration-additivity ($p$-additivity for short) is a characterizing property of the Expected Shortfall at level $p$,  as studied in \cite{WZ21} and \cite{AL24}.  Using the desired additivity properties to narrow down the potential choices of risk measures is a recurring theme in the study of risk measures; see \cite{MPST24, BKMS21, PWW25} for some recent developments. We follow this classic approach in this paper. Since any instance of partial comonotonicity is weaker than comonotonicity and the corresponding additivity properties are reversed in strength, any functional satisfying any instance of partial comonotonic additivity must be comonotonic additive. When studying the aggregation of risk, we focus on the dependence structure of the random losses, and hence law invariance is also a standard assumption. Therefore, we will be working with distortion risk measures and their non-monotone generalization, distortion riskmetrics.

In Section \ref{sec:set}, we formulate partial comonotonicity by using subsets $K$ of $[0,1]$. From this formulation, partial comonotonicity can be seen as a natural generalization of (single-point) $p$-concentration; thus, we call it $K$-concentration.   Section \ref{sec:func} develops an equivalent formulation of partial comonotonicity through
a parametrization by left-continuous increasing functions $g:[0,1]\to \R$. This formulation enables us to obtain an elegant result in Section \ref{sec:spectral}: for any random vector $\mathbf X$, a spectral risk measure with risk spectrum $g$ is additive for $\mathbf X$ precisely when partial comonotonicity, parametrized by $g$, holds for $\mathbf X$.  Section \ref{sec:app} provides two applications. First, we demonstrate the derivation of additivity conditions for several well-known functionals in the literature. Second, we illustrate how restricting to given additivity conditions, narrows down the possible choices of risk measures. Section \ref{sec:ext} extends the domain of riskmetrics beyond $L^\infty$ and  Section \ref{sec:conclusion} concludes the paper.

\section{Notations, conventions, and  some basic facts}

Throughout, increasing functions are in the weak sense,
and inequalities among vectors are componentwise. 
We typically use capital letters, such as $X$, to denote random variables. We use bold capital letters such as $\mathbf {X} $ to denote random vectors.

For real numbers $p, p_1, p_2, \dots$, we write $p_n\uparrow p$ to mean that $p_n$ is strictly increasing and  converges to $p$. Similarly, $p_n\downarrow p$ means that $p_n$ is strictly decreasing and converges to $p$. Given a function $f: \R \to \R$, we define  $f(x-)=\lim_{y\uparrow x}f(y)$ and $f(x+)=\lim_{y\downarrow x}f(y)$ when they exist.

\subsection{Risk functionals}
Fix an atomless probability space $(\Omega, \mathcal{F}, \mathbb{P})$, let $L^0$ be the collection of real-valued random variables and $L^\infty$ be the collection of essentially bounded random variables, with almost surely equal objects treated as identical. 

A risk functional $\rho: \mathcal X \to \R$  is a functional such that its domain $\mathcal X$ is a convex cone inside  $L^0$. We restrict to the case  $\mathcal X = L^\infty$ for most of this paper. A random vector $X = (X_1, \dots, X_d)$ is understood as an element of $\mathcal X^d$. 
Commonly used functionals include  the Value-at-Risk (VaR) at level $p\in (0,1)$: 
$$
\mathrm{VaR}_{p}(X)= \inf \{x \in \mathbb{R}: \mathbb{P} (X\leq x) \geq p\},
$$
 and the Expected Shortfall (ES) at level $p\in (0,1)$: 
$$
\mathrm{ES}_p(X)=\frac{1}{1-p} \int_p^1 \mathrm{VaR}_q (X) \mathrm{~d} q.
$$
We also define $\VaR_1(X)$ and $\ES_1(X)$ to be $\mathrm{ess\text{-}sup}(X)$, the essential supremum of $X$. For a fixed random variable $X$, the left quantile function $Q^-_X:(0,1]\to \R$  is defined as $Q^-_X(p)  = \VaR_p(X)$ and the right quantile function of $X$, $Q_X^+: [0,1)\to \R$ is defined as
$$Q^+_X(p) = \inf\{x\in \R: \P(X \leq x) > p\}.$$
We write the corresponding right-continuous Value-at-Risk as $\VaR^+_p(X)$. Note that at $p=0$, $Q_{X}^+(0)$ is the essential infimum.   We will be interested in real valued functions with domain $[0,1]$. So we define the function $Q_X:[0,1]\to \R$ to be $Q_X^-$ on $(0,1]$ and $Q_X^+$ at $0$. Note that $Q_X$ is right-continuous at $0$ and left-continuous on $(0,1]$.

A coherent risk measure (\cite{ADEH99}) is a functional $\rho$ satisfying the following properties:
\begin{itemize} 
    \item (M) Monotonicity: $\rho(X)\leq \rho(Y)$ for all $X,Y\in \mathcal X$ with $X\leq Y$,
    \item (TI) Translation invariance: $\rho(X+c)=\rho(X)+c$ for all $X \in \mathcal X$ and $c\in \mathbb{R}$,
    \item (SA) Subadditivity: $\rho(X+Y) \leq \rho(X)+\rho(Y)$ for all $X, Y \in \mathcal X$,
    \item (PH) Positive homogeneity: $\rho(\lambda X)=\lambda \rho(X)$ for all   $X,Y\in \mathcal X$ and $\lambda \geq 0$,
\end{itemize}
and condition (SA) may be substituted with the following:
\begin{itemize}
    \item (CX) Convexity: $\rho(\alpha X+(1-\alpha)Y)\leq \alpha \rho(X)+(1-\alpha)\rho(Y)$ for all $X, Y \in \mathcal X$ and $\alpha \in (0,1)$.
\end{itemize}

A risk functional satisfying (M) is referred to as a risk measure. If, in addition, (TI) is satisfied, the risk measure is monetary. These two properties allow us to interpret a monetary risk measure as a capital requirement.

A random vector $(X_1, \dots, X_d) \in \mathcal X^d$ is called comonotonic if there exists a random variable $Z$ and increasing functions $f_1, \dots, f_d: \R \to \R$ such that 
$$(X_1, \dots, X_d) = \left (f_1(Z), \dots, f_d(Z)\right) ~\mbox{almost~surely.}~$$

Two more commonly seen properties of risk functionals are
\begin{itemize}
     \item (LI) Law invariance: $\rho(X)=\rho(Y)$ whenever $X \laweq Y$.
     \item (CA) Comonotonic additivity: $\rho(X_1)+ \dots +\rho(X_d) = \rho(X_1 + \dots + X_d)$ whenever $(X_1, \dots, X_d)$ is comonotonic.
\end{itemize}

A spectral risk measure $\rho$ as introduced in \cite{A02} can be represented in the form 
$$\rho(X) =\rho_g(X) = \int_0^1 g(t)Q_X(t) \d t$$ where $g:[0,1]\to \R$ is non-negative, increasing and integrates to $1$. The function $g$ is referred to as a risk spectrum. Spectral risk measures are coherent and satisfy (LI) and (CA). \cite{K01} further showed that all coherent risk measures that satisfy (LI) and (CA) are of the form
$$\alpha \rho_g + (1-\alpha) \mathrm{ess\text{-}sup}$$
for some spectral risk measure $\rho_g$ and $\alpha \in [0,1]$.

A distortion riskmetric is a risk functional of the form
\begin{equation}\label{eq:DRM}
I_h(X) = \int_{-\infty}^0\left [(h(\P(X>x))-h(1)\right ]\d x + \int_0^\infty h(\P(X>x))\d x
\end{equation} 
where $h:[0,1]\to \R$, called the distortion function, is of bounded variation with $h(0)=0$. We denote by  $\mathcal H^{\rm bv}$ the set of all distortion functions.  
Distortion riskmetrics are characterized in \cite{WWW20} as risk functionals satisfying (LI) and (CA) together with the following continuity condition:
\begin{itemize}
     \item (UNC) Uniform norm-continuity: $\rho$ is uniformly continuous with respect to $L^\infty$-norm.
\end{itemize}

The conjugate function of a distortion function $h$ is defined as $\hat h(t) = h(1)-h(1-t)$. If $h$ is increasing with $h(1) = 1$ and $\hat h$  is differentiable with an increasing derivative $g$, we have $I_h = \rho_g$ so that spectral risk measures are a special case of distortion riskmetrics. 
The following facts from \cite{WWW20b} about distortion riskmetrics will be useful.
\begin{fact}\label{fact:choquet}
    Let $h\in \mathcal H^{\rm bv}$  and $X\in L^\infty$.
\begin{enumerate}[(a)]
    \item \label{choquet:rquantile} if $h$ is right-continuous, then $I_h(X) = \int_0^1 Q_X^+(1-p)\d h(p)$;
    \item \label{choquet:lquantile} if $h$ is left-continuous, then $I_h(X) = \int_0^1 Q_X(1-p)\d h(p)$;
    \item $I_h$ satisfies (PH);
    \item $I_h$  satisfies (CX) if and only if $h$ is concave.
\end{enumerate}
\end{fact}

\subsection{Tail events and concentration}

An event $A\in \mathcal F$ is called a tail event for a random variable $X$ if
$$X(\omega)\geq X(\omega^\prime) \mbox{~holds a.s. for all }\omega \in A, ~\omega^\prime \in A^c, $$
where $A^c$ stands for the complement of $A$. 
If a tail event $A$ has $\P(A)=1-p$ for $p\in [0,1]$, then $A$ is called a $p$-tail event. A random vector is said to be $p$-concentrated if all components share a common $p$-tail event.
We collect the following useful facts for $p$-concentration from \cite{WZ21}.
\begin{fact}\label{fact:pcon}
The following properties of  $p$-concentration hold.
\begin{enumerate}[(a)]
    \item \label{lem:com} A random vector is $q$-concentrated for all $q\in (0,1)$ if and only if it is comonotonic.
    \item \label{lem:sum}  Suppose that  $ (X_1, \dots, X_d)$  is $p$-concentrated for some $p\in(0,1)$. A set $A$ is a $p$-tail event of $ X_1 + \dots + X_d$ if and only if $A$ is a $p$-tail event for each $X_i$ for $i= 1, \dots, d$;
    \item \label{lem:ineq}   Let $p\in (0,1)$, if $  (X_1, \dots, X_d)$  is $p$-concentrated, then we have the inequalities
    $$\VaR_{p} \left(\sum_{i=1}^d X_i\right)\leq \sum_{i=1}^d \VaR_{p} (X_i) \leq \sum_{i=1}^d \VaR^+_{p} (X_i) \leq \VaR^+_{p} \left(\sum_{i=1}^d X_i\right).
    $$
\end{enumerate}
\end{fact}


\section{Set-indexed partial comonotonicity}\label{sec:set}

\subsection{Defining $K$-concentration}
The first of the two formulations of partial comonotonicity comes as a generalization of $p$-concentration. 
\begin{definition}
    Given $K\subseteq[0,1]$, we say a random vector $\mathbf{X}$ is \textit{$K$-concentrated},  $K$-con in short, if  $\mathbf{X}$ is $p$-concentrated for each $p\in K$.
\end{definition}
 
 Because any random vector $\mathbf X$ is trivially $1$-concentrated and $0$-concentrated, we sometimes omit $0$ and/or $1$ from the set $K$.  For $K_1, K_2 \subseteq [0,1]$, we sometimes abuse notation and write $K_1=K_2$ to mean $K_1\cap (0,1) = K_2 \cap (0,1)$.  In this paper, the word dependence structure will refer to either $K$-concentration for some $K\subseteq[0,1]$, or a particular copula (\cite{N06}). 

\begin{example}\label{ex:1}  By choosing the set $K$, we can recover many well-studied dependence structures. 
\begin{enumerate}[ (i)]
     \item For $p\in (0,1)$, $\{p\}$-concentration is precisely $p$-concentration. This is the most distinguishing property used to characterize $\ES_p$ in \cite{WZ21}. 
     \item With $K_c = [0,1]$,  $K_c$-concentration is equivalent to  comonotonicity (Fact \ref{fact:pcon}, item (\ref{lem:com})). This justifies the term partial comonotonicity since $K$ is always a subset of $[0,1]$. 
 
    \item For $p\in (0,1)$, let $K_p = [p,1]$. For $\mathbf X$ to be $K_p$-concentrated, 
    there exists a common $p$-tail event of $\mathbf X$ such that the conditional distribution of $\mathbf X$ on this event is comonotonic. This was called upper comonotonicity in \cite{C09}.
    
    \item For $p\in (0,1)$, take any strictly increasing sequence $\{p_n\}_{n\in \mathbb N}$ in $(0,1)$ such that $p_n \uparrow p$. Let $K = \{p_n\} \cup \{p\}$. To the best of the author's knowledge, such dependence structures had not been mentioned in the literature. As we will see in Theorem \ref{thm:quant}, this dependence structure guarantees the additivity of $\VaR_p$.

\end{enumerate}    
\end{example}

We first need a structured approach to selecting the common $p$-tail events for different levels $ p \in K$.
\begin{proposition}

    Let $K\in \mathcal S$ and suppose a random vector $\mathbf X$ is $K$-concentrated. Then there exists a collection of sets $\{A_p\}_{p\in K}$ such that 
    \begin{enumerate}[(i)]
        \item for each $p\in K$, $A_p$ is a common $p$-tail event for $\mathbf X$;
        \item \label{cond:2} for $p,q \in K$ with $p > q$, we have $A_p \subseteq A_q$. 
    \end{enumerate}
\end{proposition}

\begin{proof}
    Let $S = X_1 + \dots + X_n$. By Lemma A.32 of \cite{FS16}, there exists a uniform $[0,1]$ random variable $U$ such that $Q_S(U)=S$ almost surely. For each $p\in K$, let $A_p = \{U > p\}$. Clearly, each $A_p$ is a $p$-tail event for $S$ and the collection satisfies condition (\ref{cond:2}). By item (\ref{lem:sum}) of Fact \ref{fact:pcon}, each $A_p$ is also a common $p$-tail event for $\mathbf X$.
\end{proof}

    We call $\{A_p\}_{p\in K}$ satisfying (i) and (ii) above a nested collection of common tail events. 
The next proposition shows that it suffices to consider concentration indexed by closed subsets of $[0,1]$.

\begin{proposition}

    For $K\subseteq [0,1]$
    with its closure denoted by $\mathsf{cl}(K)$, $K$-concentration is equivalent to $\mathsf{cl}(K)$-concentration. 
\end{proposition}

\begin{proof}
     If $p_n \downarrow p$ with each $p_n \in K$,  let $\{A_p\}_{p\in K}$ be a nested collection of common tail events and $A = \bigcup_{n \in \mathbb N} A_{p_n}$.  Since $\P$ satisfies $\sigma$-additivity, we have $\P(A) = 1-p$. If $\omega \in A$, then $\omega \in A_{p_n}$ for some $n$. Since $A^c \subseteq A_{p_n}^c$, for any $\omega^\prime \in A^c$ we have $\mathbf X(\omega) \geq \mathbf X(\omega^\prime)$, and hence $A$ is a $p$-tail event for $\mathbf X$.  Similarly, if $p_n \uparrow p$ with each $p_n \in K$,  $\mathbf X$ is  $p$-concentrated. 
\end{proof}

Let $\mathcal S$ denote  the set of all closed subsets of $[0,1]$.  Two partial comonotonicities are \textit{equivalent} if either one implies the other. For $K_1, K_2\in \mathcal S$, if $K_1$-concentration is equivalent to $K_2$-concentration, we write $K_1$-con $\cong K_2$-con. If $K_1=K_2$ in $\mathcal S$, clearly $K_1$-con $\cong K_2$-con. The converse will be established in Corollary \ref{cor:1}.

\subsection{Partial comonotonic additivity}
We are ready to develop the notion of partial comonotonic additivity.
\begin{definition}
    Let $K\in \mathcal S$. A risk functional $\rho$ is \textit{$K$-additive} if for all $K$-concentrated random vectors $\mathbf X = (X_1, \dots, X_d) $, we have 
    $$\rho\left(\sum_{i=1}^d X_i\right) = \sum_{i=1}^d \rho( X_i).$$
\end{definition}

We begin by identifying the sets $K\in \mathcal S$ such that the left and right quantile functions are $K$-additive.

\begin{theorem} \label{thm:quant}
    Let $p\in (0,1]$ and $K$ be a closed subset of $[0,1]$. The following are equivalent:
    \begin{enumerate}[(i)]
        \item $\mathrm{VaR}_p$ is $K$-additive.
        \item There exists a strictly increasing sequence $p_n \uparrow p$ with each $p_n \in K$.
   
    \hspace*{-2.7em} Similarly, for $p\in [0,1)$, the following are equivalent:
   
        \item $\mathrm{VaR}^+_p$ is $K$-additive.
        \item There exists a strictly decreasing sequence $p_n \downarrow p$ with each $p_n \in K$.
    \end{enumerate}
\end{theorem}

\begin{proof}
    Let $p_n \uparrow p$ and $\mathbf X = (X_1, \dots, X_d)$ be $p_n$-concentrated for each $n\in \mathbb N$.  By (\ref{lem:ineq}) from Fact \ref{fact:pcon}, for each $n\in \mathbb N$ we have the inequalities
    \begin{align*} \sum_{i=1}^d \VaR_{p_{n}} (X_i) \leq \sum_{i=1}^d \VaR^+_{p_{n}} (X_i) &\leq \VaR^+_{p_{n}} \left(\sum_{i=1}^d X_i\right)\\
    &\leq \VaR_{p_{n+1}} \left(\sum_{i=1}^d X_i\right ) \leq \sum_{i=1}^d \VaR_{p_{n+1}}(X_i).
    \end{align*}
As $n\to \infty$, by the left-continuity of the left quantile function, the first and last terms converge to $\sum_{i=1}^d {\VaR_p(X_i)}$,  and therefore the last two terms are equal in the limit.\\
Similarly for $p_n \downarrow p$, we have
    \begin{align*} \sum_{i=1}^d \VaR^+_{p_{n}} (X_i) \geq \sum_{i=1}^d \VaR_{p_{n}} (X_i) &\geq \VaR_{p_{n}} \left (\sum_{i=1}^d X_i\right )\\
    &\geq \VaR^+_{p_{n+1}} \left (\sum_{i=1}^d X_i\right ) \geq \sum_{i=1}^d \VaR^+_{p_{n+1}}(X_i).
    \end{align*}
    As $n\to \infty$, by the right-continuity of the right quantile function, the first and last terms converge to $\sum_{i=1}^d {\VaR^+_p(X_i)}$ and the last two terms are equal in the limit.\\
    For the converse directions, see the proof of Theorem \ref{thm:KCA}.
\end{proof}

The following proposition enables us to focus on the special case $d=2$ for many of the proofs that follow.

\begin{proposition}
\label{prop:incr-trans}
Fix $d\in \mathbb N$, $K\in \mathcal S$ and $\mathbf X\in \mathcal X^d$. We have that 
    $\mathbf X  $ is $K$-concentrated if and only if for any $k\in\mathbb N$ and any component-wise  increasing function $\mathbf f = (f_1, \dots, f_k)\footnote{For each $i=1,\dots, k$, $f_i$ is a function from $\R^d$ to $\R$.} : \R^d \to \R^k$, $\mathbf f (\mathbf X)$ is $K$-concentrated.
\end{proposition}
\begin{proof}
    Take $p \in K$ and let $A_p$ be a common $p$-tail event for $\mathbf X$. For a.s. all $\omega \in A_p$ and  $\omega^\prime \in A_p^c$, we have $\mathbf X(\omega) \geq \mathbf X(\omega^\prime)$. Since $\mathbf f$ is increasing, $[\mathbf f(\mathbf X)](\omega) \geq [\mathbf f(\mathbf X)](\omega^\prime)$ so that $A_p$ is a common $p$-tail event for $\mathbf f(\mathbf X)$, and hence $\mathbf f(\mathbf x)$ is $p$-concentrated. 
\end{proof}

The invariance property under increasing transforms in Proposition \ref{prop:incr-trans}  is also satisfied by the popular dependence concept of positive association, introduced by \cite{EPW67}, who showed that 
if $\mathbf X$ is positively associated, then so is $\mathbf f(\mathbf X)$ for any increasing function $\mathbf f $. Proposition \ref{prop:incr-trans} generalizes
the fact that comonotonicity satisfies this invariance,  
 shown by \cite{LLW23}.

\begin{remark}
    Suppose $\rho$ is additive for $2$-dimensional $K$-concentrated vectors. That is, $\rho(X+Y) = \rho(X)+\rho(Y)$ for any $K$-concentrated random vector $(X,Y)$. Since addition is a pointwise increasing function, for  a $K$-concentrated random vector $(X_1, X_2, X_3)$, $(X_1, X_2+X_3)$ and $(X_2, X_3)$ are both $K$-concentrated. So that
\begin{align*}
    \rho(X_1+X_2+X_3) = \rho(X_1+(X_2+X_3)) &= \rho(X_1)+\rho(X_2+X_3) \\ &= \rho(X_1)+\rho(X_2)+\rho(X_3).
\end{align*}
This shows that considering the case $d=2$ does not lose generality. 
\end{remark} 

We are ready to develop criteria for $K$-additivity for general distortion riskmetrics.
For any  $K\in \mathcal S$, we decompose $  K^c$ as a disjoint union of open intervals $\{B_i\}_{i \in \mathbb N}$.
For a subset $B$ of $[0,1]$, let $1-B=\{1-x:x\in B\}$.

\begin{theorem}\label{thm:KCA} 
   For   $K\in \mathcal S$ and $h\in \mathcal{H}^{\rm bv}$,  a distortion riskmetric $I_h$ is $K$-additive if and only if
   $h$ is linear on each $\mathsf{cl}(1-B_i$).
\end{theorem}  

\begin{proof}
    Given $h$ with bounded variation, decompose $h = h_l+h_r+h_c$ where $h_l$ (resp. $h_r$) is constant except for jumps that are left- (resp. right-) continuous,  and $h_c$ is continuous. Let $\hat h, \hat h_c, \hat h_l, \hat h_r$ be corresponding conjugate functions. The functional $I_{h_l}$ is a countable weighted sum of left quantiles, and similarly, $I_{h_r}$ is a countable weighted sum of right quantiles. Note that $h$ being continuous on $\mathsf{cl} (1-B_i)$ is equivalent to $\hat h$ being continuous on $\mathsf{cl}(B_i)$.

By Fact \ref{fact:choquet} (\ref{choquet:rquantile}) and (\ref{choquet:lquantile}),  we have
\begin{align*}
    I_h &= I_{h_c}+I_{h_l}+I_{h_r}\\
    &=\int_0^1 Q^-_{(~\cdot~)}(p) \d \hat h_c (p) +\int_0^1 Q^-_{(~\cdot~)} (p)\d \hat h_l (p)+\int_0^1 Q^+_{(~\cdot~)} (p)\d \hat h_r (p).
\end{align*}
Suppose $(X,Y)$ is $K$-concentrated. Let $\{A_p\}_{p\in K}$ be a nested family of common tails, we have 
\begin{align*}
I_{h_c}(X+Y) &= \int_0^1 Q^-_{X+Y}(p) \d \hat h_c(p)\\
&= \sum_{i=1}^\infty \int_{K_i} Q^-_{X+Y}(p) \d \hat h_c(p) + \sum_{i=1}^\infty \int_{B_i} Q^-_{X+Y}(p) \d \hat h_c(p).
\end{align*}
Let $m_i$ be the slope of $\hat h_c$ on $B_i=(a_i,b_i)$, and let $\Omega_i = A_{a_i} \setminus A_{b_i}$, we have:
\begin{align*}
     \int_{B_i} Q^-_{X+Y}(p) \d \hat h_c(p) 
    &= m_i \E[(X+Y) \cdot \mathds{1}_{\Omega_i}]\\ 
    &= m_i \E[X \cdot \mathds{1}_{\Omega_i}] +m_i \E[Y \cdot \mathds{1}_{\Omega_i}]\\
    &= \int_{B_i} Q^-_{X}(p) \d \hat h(p) +\int_{B_i} Q^-_{Y}(p) \d \hat h(p).
\end{align*} 
 For an interior point $p$ of $K$, Theorem \ref{thm:quant} implies that $Q^-_{X+Y}(p) = Q^-_X(p)+Q^-_Y(p)$. Combining the above arguments,  $I_{h_c}$ is $K$-additive.

Next suppose for some $p\in [0,1]$, $\VaR_p^-$ is not $K$-additive. By Theorem \ref{thm:quant}, $p \in (a_i, b_i]$ for some $B_i = (a_i, b_i)$. Since $\hat h$ is continuous on $[a_i,b_i]$, $\hat h_l$ is continuous at $p$ and $I_{h_l}$ does not have a mass on $\VaR_p^-$, this shows that $I_{h_l}$ is $K$-additive. A similar argument shows that $I_{h_r}$ is also $K$-additive.

For the only if direction, we construct counterexamples taking $\Omega = [0,1]$.

Case 1: If $\hat h$ is not linear on some $B_i = (a,b)$. Pick $X,Y$ such that $X=Y = 0$ on $[0,a]$ and $X = Y = 3$ on $[b,1]$. Since $I_h$ is not a multiple of the expectation on $(a,b)$, we can pick $X, Y$ to have range $[1,2]$ on $(a,b)$ so that $I_h$ restricted to $B_i$ is not additive. This can be done because $K$-concentration puts no restriction on the dependence structure of ($X|_{(a,b)}, Y|_{(a,b)})$.  Since $X=Y$ outside $(a,b)$, $I_h$ is additive when restricted to $(a,b)^c$. Overall, $I_h$ is not additive for $(X,Y)$. 

Case 2: Suppose for some $B_i = [a_i, b_i]$,  $\hat h$ is not right-continuous at $a$. Then $I_h$ puts a mass on $Q^-_a$ and for some $\epsilon \in (0,a)$,  $(a-\epsilon, a) \cap K = \emptyset$. A counterexample is illustrated by Figure \ref{fig:example}. Let $U$ be a uniform $[0,1]$ random variable, $X = f(U)$ and $Y=g(U)$ with $f,g$ being the two functions in the illustration. 

Case 3: A similar counterexample can be constructed where additivity fails due to a mass on some $Q^+_{b_i}$.  
\end{proof}
\begin{figure}
    \centering
    \begin{center}
    \begin{tikzpicture}
\draw[<->] (0,3.4) -- (0,0) -- (3.4,0);
\draw[gray,dotted] (0,0.6) -- (3,0.6);
\draw[gray,dotted] (0,1.2) -- (3,1.2);
\draw[gray,dotted] (0,1.8) -- (3,1.8);
\draw[gray,dotted] (0,2.4) -- (3,2.4);
\draw[gray,dotted] (0.6,0) -- (0.6,3);
\draw[gray,dotted] (1.2,0) -- (1.2,3);
\draw[gray,dotted] (1.8,0) -- (1.8,3);
\draw[gray,dotted] (2.4,0) -- (2.4,3);
\draw[gray,dotted] (3,0) -- (3,3);
\draw[gray,dotted] (0,3) -- (3,3); 

\node[below] at (3,0) {$1$}; 
\node[below] at (0,0) {$0$}; 
\node[below] at (1,0) {$a-\epsilon$}; 
\node[below] at (2,0) {$a$}; 

\draw[blue] (0,1) -- (1, 1);
\draw[red] (0,0.9) -- (1, 0.9);

\draw[blue] (2,2) -- (3, 2);
\draw[red] (2,2.1) -- (3, 2.1);

\draw[blue] (1,1) -- (2, 2);
\draw[red] (1,2) -- (2, 1);
    \end{tikzpicture} 
\end{center}   
    \caption{An example where $Q_a^-$ and $\rho$ are not additive}
    \label{fig:example}
\end{figure}
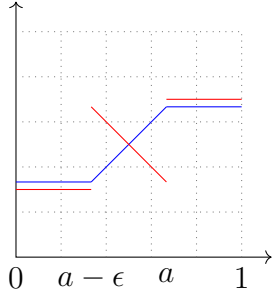

\begin{example}We exhibit some explicit counter-examples for each of the three cases from the above proof. Fix $\Omega = [0,1]$. Take $X$ and $Y$ to be $0$ on $[0, 2/3]$ and $3$ on $[5/6,1]$. On $(2/3, 5/6)$, let $X(t) = 9t-4.5$ and $Y(t) = -9t+7.5$. Note that $X \laweq Y$.

Case 1: Let $K = [0,1/4] \cup [3/4, 1]$ and $g(t) = X(t) = (9t-4.5)_+ \wedge 3$ for $t\in [0,1]$. That is, $g$ is $0$ on $[0,2/3]$, increases linearly on $[2/3, 5/6]$ and is $3$ on $[5/6, 1]$,  so that the corresponding $\hat h$ for $\rho_g$ is not linear on $[2/3, 5/6]$. The complement $K^c$ has one connected component $(3/8, 5/8)$, on which $\hat h$ is not linear.  We have 
$$(X+Y)(t) = 3 \cdot\mathds{1}\{(2/3, 5/6)\} + 6 \cdot \mathds{1}\{[5/6,1]\}.$$
and we compute that
\begin{align*}
    \rho_g(X) &= 0 + \int_{2/3}^{5/6} (9t-4.5)^2 \d t + \int_{5/6}^1 3^2 \d t = 2.375\\
    \rho_g(X+Y) &= 0 + \int_{2/3}^{5/6} 3\cdot (9t-4.5) \d t + \int_{5/6}^1 6\cdot 3 \d t = 4.125.
\end{align*}
    Therefore, $\rho_g(X+Y) < 2\rho_g(X) = \rho_g(X)+\rho_g(Y)$.

    Case 2: Let $\rho = \VaR^-_{5/6}$. We have that $\rho(X) = \rho(Y) = \rho(X+Y) = 3$, so that $\rho_g(X+Y) < \rho_g(X) + \rho_g(Y)$.

    Case 3: Let $\rho = \VaR^+_{2/3}$. We have that $\rho(X) = \rho(Y) = 0$ and $\rho(X+Y) = 3$. Again, $\rho$ is not additive on $(X,Y)$.
\end{example}

\begin{remark}
    The following relationships are obvious from the definitions. Let $K_1 \subseteq K_2 $ be closed subsets of  $[0,1]$: 
\begin{enumerate}[(i)]
    \item for any random vector, $K_2$-concentration implies $K_1$-concentration;
    \item for any distortion riskmetric,  $K_1$-additivity implies  $K_2$-additivity.
\end{enumerate}
\end{remark}

\section{Function-indexed partial comonotonicity}\label{sec:func}
\subsection{Defining $g$-comonotonicity}
For random variables $X,Y,Z$, if $(X,Z)$ and $(Y,Z)$ are comonotonic, it is in general not true that $(X,Y)$ is also comonotonic. That is, the property of having perfect comovement is not transitive.  However, depending on $Z$, the property is partially transitive.

In this section, we define and investigate partial comonotonicity indexed by functions, and establish an equivalence to set-indexed comonotonicity. 

Let $\mathcal G$ be the collection of  increasing functions from $[0,1]$ to $\R$ that are right-continuous at $0$ and left-continuous on $(0,1]$ Note that for any increasing function $f: [0,1]\to \R$, there is a unique function  $g\in \mathcal G$ such that $f=g$ almost everywhere, given by $g(0) = f(0+)$ and $g(p) = f(p-)$ for all $p\in (0,1]$. We write $g = \mathsf{LC}(f)$.   
\begin{definition}
Let $g \in \mathcal G$. A random vector $(X_1, \dots, X_d)$ is \textit{$g$-comonotonic} ,  or $g$-com in short,   if there exists a random variable $Z$ such that
\begin{enumerate}[(i)]
    \item the left quantile function $Q_Z$ is equal to $g$, and
    \item $(X_i,Z)$ is comonotonic for all $i=1, \dots, d$.
\end{enumerate}
\end{definition}

\begin{definition}
    Given $g \in \mathcal G$, a risk functional $\rho$ is
 \textit{$g$-additive} if for all random vector $(X_1, \dots, X_d)$ satisfying $g$-comonotonicity, we have $\rho(X_1 + \cdots +X_d) = \rho(X_1)+ \cdots + \rho(X_d)$.

\end{definition}

\subsection{Relative strengths of function-indexed partial comonotonicity}


We investigate the relative strengths of partial comonotonicity indexed by two functions $f,g\in \mathcal G$.  Given $g \in \mathcal G$, we say $p\in [0,1)$ is a \textit{point of strict increase on the right} (PSIR) if $g(p)<g(s)$ for all $s\in (p,1]$, \textit{point of strict increase on the left} (PSIL) is defined similarly. A point $p\in [0,1]$ is a \textit{point of strict increase} (PSI) if it is either a PSIR or a PSIL (or both). We write $\mathsf{PSI}(g)$ for the set of PSIs of $g$.
\begin{lemma} \label{lem:order}

    For $f,g \in \mathcal G$, the following are equivalent:  
    \begin{enumerate}[(a)]
    \item there exists an increasing function $h$ such that $f = h \circ g$ almost everywhere; 
    \item there exists an increasing function $h$ such that $f = \mathsf{LC}(h \circ g)$; 
    \item\label{L1c} for all $p,q \in [0,1]$ with $p<q$ such that $g(p)=q(p)$. Then either $f(p)=f(q)$ or \begin{equation}\label{eq:psil}
        p \text{~is a PSIL for~} g \text{~and~}g\text{~is~constant~on~}(p,q].
    \end{equation}
    \item There is a subset $A$ of $[0,1]$ with Lebesgue measure $1$ such that for all $p,q \in A$, $$g(p)=g(q) \implies f(p)=f(q).$$
    \end{enumerate}
    
\end{lemma}

\begin{proof}
     
(a)$\Rightarrow$(b): Let $A\subseteq [0,1]$ be the set with Lesbesgue measure $1$ such that $f=h\circ g$ on $A$. Since $A^c$ has Lebesgue measure $0$, it does not contain any open intervals. For any $p\in (0,1]$ and $p > \epsilon_n > 0$, there exists some $p_n\in (p-\epsilon_n, p] \cap A$. Taking $\epsilon_n \to 0$, we have a sequence $p_n \uparrow p$ with each $p_n \in A$. Hence 
$$f(p) = f(p-) = \lim_{n\to \infty} f(p_n) = \lim_{n\to \infty} h\circ g(p_n) = h\circ g(p-) = \mathsf{LC}(h\circ g)(p).$$
Similarly, for $p=0$, there is a sequence $p_n\downarrow p$ with each $p_n \in A$, so that $f(0) = f(0+) = (h\circ g)(0+) = \mathsf{LC}(h\circ g)$.

(b)$\Rightarrow$(c): 
Let $p<q$ be such that $g(p)= g(q)$. Suppose that Condition \eqref{eq:psil} fails. Clearly $g$ is constant on $[p,q]$, so that $p$ is not a PSIL for $g$. Then $g$ is constant on some interval $[p-\epsilon, q]$. Hence $h\circ g$ is left-continuous at $p$ and $q$, so that $f(p) = f(q)$.

(c)$\Rightarrow$(d): 
The set $B$ of $p$ such that there exists some $q>p $ satisfying Condition $\eqref{eq:psil}$ is at most countable, and has Lebesgue measure $0$. This is because the collection of intervals $(p,q)$ are nonempty and disjoint. Let $A$ be the complement of $B$, we have that (d) holds.

(d)$\Rightarrow$(a):
For each $p\in A$, set $h(g(p)) = f(p)$, so that $h$ is defined on $B$, the range of $g$. This is well-defined because if $g(p) = g(q)$, then $h(g(p)) = f(p) = f(q) = h(g(q))$. If $a<b$ in $B$, there exists $p<q$ such that $g(p)=a$ and $g(q) = b$, hence $h(a) = f(p) < f(q) = h(b)$ so that $h$ is increasing. 
We have $g = h\circ f$ on $A$, proving (a). 
\end{proof}
For $f,g\in \mathcal G$, we say $f \precsim g$ if any of the above equivalent conditions hold.
\begin{lemma}
      \label{lem:lc}
    Let $h$ be an increasing function and $X\in \mathcal X$. Then we have
    $$Q_{h(X)}(p) = h\left (Q_X(p)\right) ~\mbox{almost everywhere}.$$
    In particular, 
    $$Q_{h(X)}(p) = \mathsf{LC}\left(h\left (Q_X(p)\right)\right) ~\mbox{for~all~}p\in(0,1].$$
\end{lemma}
 
\begin{proof}
    The first statement is Lemma A.27 from \cite{FS16}. Since $Q_{h(X)}(p)$ is left-continuous, it is fixed by $\mathsf{LC}$. Taking $\mathsf{LC}$ on both sides yields the second statement. 
\end{proof}

\begin{proposition}\label{prop:order}
    Given $f,g \in \mathcal G$, we have $f \precsim g$ if and only if  $g$-comonotonicity implies $f$-comonotonicity.
\end{proposition}

\begin{proof}
    ($\Rightarrow$): Suppose $(X_1, \dots, X_d)$ is $g$-comonotonic, let $Z$ be such that $Q_Z^- = g$ and $X_i$ comonotonic to $g$ for each $i = 1, \dots, d$. Let $h$ be such that  $f = \mathsf{LC}(h \circ g)$.
    
    We show that $h(Z)$ satisfies the definition of $f$-comonotonicity for the random vector $(X_1, \dots, X_d)$. By Lemma \ref{lem:lc},  we have $Q_{h(Z)} = \mathsf{LC}(h\circ g)$. For each $i=1, \dots, d$, $(X_i, Z)$ being comonotonic clearly implies the comonotonicity of $(X_i, h(Z))$.

    ($\Leftarrow$): Suppose $f \precnsim g$. We use formulation \eqref{L1c} of Lemma \ref{lem:order}, so that there exists $p<q$ with $g(p)=g(q)$ and $f(p)\neq f(q)$, and $p$ is not a PSIL for $g$. Hence there exists $p'<p$ such that $g$ is constant on $[p',q]$. Let $q' = \sup \{ t: f(t) = f(p)\}$. By left-continuity of $f$, $q' < q$. We have that $f \leq f(p)$ on the non-degenerate interval $[p', p]$ and $f>f(p)$ on the non-degenerate interval $[q', q]$.
    Define the function $h:[0,1] \to [0,1]$ as follows:
    $$h(t) = \begin{cases}
        t & \mbox{~if~}t < p'\\
        q-(t-p') & \mbox{~if~} t\in [p', q]\\
        t & \mbox{~if~}t > q
    \end{cases}.$$
    Let $U$ be a uniform $[0,1]$ random variable. 
    
    We claim that $X = f(U)$ and $Y = (f\circ h)(U)$ are $g$-comonotonic but not $f$-comonotonic. Let $Z = g(U)$. Since $X, Z$ are increasing functions of $U$, they are comonotonic. Since $g$ is constant on $[p', q]$, we also have that $Z = (g\circ h)(U)$, so that $Y,Z$ are increasing functions of $h(U)$. Therefore $Z$ witnesses the $g$-comonotonicity of $(X,Y)$. 

    To see that $X,Y$ is not $f$-comonotonic, let $Z$ be any random variable with quantile function $f$. Any $q'$ tail event of $X$ must equal to the set $\Omega \supseteq B = \{ U > q'\}$ with probability $1$. Comonotonicity of $(X,Z)$ implies that $Z$ also have $B$ as its $q'$-tail event. However by construction, $Y$ cannot have $B$ as its $q'$-tail event, so that $Y,Z$ cannot be comonotonic, and hence $(X,Y)$ is not $f$-comonotonic.
\end{proof}

\begin{proposition}\label{prop:5}
    Given $f,g \in \mathcal G$, we have $f \precsim g$ if and only if $\mathsf{PSI}(f) \subseteq \mathsf{PSI}(g)$. 
\end{proposition}
\begin{proof}
    ($\Rightarrow$): Suppose $p\in \mathsf{PSI}(f)$. If $p$ is a PSIL for $f$, we show that $p$ is a PSIL for $g$. Suppose not, there exists some $q<p$ such that $g$ is constant on $[q,p]$ and hence $h\circ g$ is constant on $[q,p]$, so that $f = \mathsf{LC}(h\circ g)$ is constant on $(q,p]$, contradiction.

    Next, consider the case $p$ is a PSIR for $f$. If $p$ is a PSIR for $g$, $p\in \mathsf{PSI}(g)$. Otherwise, there exists some $q>p$ such that $g(p) = g(q)$. By Condition \eqref{L1c} of Lemma \ref{lem:order}, since $f(p)<f(q)$, we must have that $p$ is a PSIL for $g$. 

    ($\Leftarrow$): Suppose $f \precnsim g$. As constructed in the proof of Proposition \ref{prop:order}, part   ($\Leftarrow$): $q'$ is an element of $\mathsf{PSI}(f)$ but not an element of $\mathsf{PSI}(g)$.
\end{proof}
 If $f$-comonotonicity is equivalent to $g$-comonotonicity, we similarly write $f$-com $\cong g$-com.
\begin{remark}\label{rmk:3}
    The symmetric part of the pre-order $\precsim$ defines an equivalence relation $\cong$ on $\mathcal G$. By Proposition \ref{prop:order} and \ref{prop:5}, we have $$ f\mbox{-com} \cong g\mbox{-com}\iff f\cong g \iff \mathsf{PSI}(f) = \mathsf{PSI}(g).$$
\end{remark}

 \subsection{Equivalence to set-indexed partial comonotonicity}
 To establish a connection between the two notions of partial comonotonicity, we define a pair of functions between $\mathcal S$ and $\mathcal G$.
We define $\mathcal I: \mathcal G \to \mathcal S$ by 
\begin{equation*}
    \mathcal I(g) = \mathsf{PSI}(g)
\end{equation*} 
and $\mathcal V: \mathcal S \to \mathcal G$, $K \mapsto \mathcal V(K) = g$ by 
\begin{equation}\label{eq:g}
    g(p) = \sup \{ (0,p)\cap K\} \mathrm{~for~}p\in [0,1], \mathrm{~with~}\sup \emptyset := 0.
\end{equation}
We check that the maps $\mathcal I$ and $\mathcal V$ have the correct codomains, starting with $\mathcal I$.  
  \begin{lemma}
      Given $g\in \mathcal G$, $\mathcal I(g)$ is closed, so that $\mathcal I(g)\in \mathcal S$.
  \end{lemma}
  \begin{proof}
      Suppose $p_n\downarrow p$ with each $p_n \in \mathcal I(g)$. By taking a subsequence if necessary, we may assume without loss of generality that either every $p_n$ is a  PSIR or every $p_n$ is a PSIL.

      Suppose each $p_n$ is a PSIR. Given any $q>p$, there is some $p_n \in (p,q)$. We have
      $g(p)\leq g(p_n)< g(q)$ so that $p$ is also a PSIR.

      Suppose each $p_n$ is a PSIL. Given any $q>p$, there is some $p_n\in (p,q)$. We have 
      $g(p) <g(p_n) \leq p(q)$ and hence $p$ is a PSIR.

      A similar argument shows that if $p_n\uparrow p$ with each $p_n \in \mathcal I(g)$, $p\in \mathcal I(g)$.
  \end{proof}

 Given $K\in \mathcal S$, to see that $\mathcal V(K) = g$ is indeed in $\mathcal G$, if $0 \leq p <q\leq 1$, $g(q)$ is the supremum over a larger set and hence $g(p)\leq g(q)$, so that $g$ is increasing. The left-continuity is of $g$ is more involved, we start with a lemma.
\begin{lemma}\label{lem:1}
    Let $K\in \mathcal S$ and $p\in (0,1]$. Suppose there exists a strictly increasing sequence $\{p_n\}_{n\in \mathbb N}$ in $K$ such that $p_n \uparrow p$. Then the function $g$ defined by equation \eqref{eq:g} satisfies $g(p) = p$.
\end{lemma}
\begin{proof}
    Since $p_n \in (0,p_{n+1}\cap K)$, we have that $g(p_{n+1})\geq p_n$. By the increasingness of $g$, we have $g(p) \geq g(p_{n+1})$ for all $n$. Therefore 
    $$g(p) \geq \lim_{n\to \infty} g(p_{n+1}) = \lim_{n \to \infty} p_n = p.$$
    But we also have $g(p)\leq \sup\{(0,p)\}=p$, so that $g(p)=p$. 
\end{proof}
We are ready to prove the left-continuity of $\mathcal V(K)$.
\begin{lemma}
    Given $K\in \mathcal S$, the function $g=\mathcal V(K)$ defined by equation \eqref{eq:g} is left-continuous.
\end{lemma}
\begin{proof}
    Let $p\in (0,1]$ be arbitrary, we show that $g$ is left-continuous at $p$. \\
    Case 1: suppose there exists $\epsilon > 0$ such that $[p-\epsilon, p)\cap K = \emptyset$. Hence $(0,q)\cap K = (0,p-\epsilon)\cap K$ for all $q\in [p-\epsilon, p]$. This shows that $g$ is constant on $[p-\epsilon, p]$ and is hence left-continuous at $p$.\\
    Case 2: if no such epsilon exists, there exists a strictly increasing sequence $\{p_n\}_{n\in \mathbb N}$ in $K$ such that $p_n \uparrow p$. By Lemma \ref{lem:1}, $g(p) = p$. For any sequence $q_n \uparrow p$, for each $n$, exists some $k(n)$ such that $p_n < q_{k(n)}$, so that 
    $\lim_{k\to \infty}g(q_k) \geq g(p_n)$ for any $n\in \mathbb N$, and hence $\lim_{k\to \infty}g(q_k) \geq g(p)$. But each $g(q_k) \leq g(p)$, we have that $g$ is left-continuous at $p$.
\end{proof}

\begin{proposition}\label{prop:adj} 
 
The maps $\mathcal V$ and $\mathcal I$ have the following properties:
    \begin{enumerate}[(a) ]
        \item\label{6a} $\mathcal I (\mathcal V(K)) = K$ for all $K\in \mathcal S$.
        \item\label{6b} $\mathcal V (\mathcal I(g)) \cong g$ for all $g \in \mathcal G$; 
    \end{enumerate}
\end{proposition}

\begin{proof}
    \begin{enumerate}[(a)]
        \item Let $g = \mathcal V(K)$. We first show that $K \subseteq \mathcal I(g)$. Fix $p\in K$, since $K$ is closed, either there is a sequence $p_n\uparrow p$ or there is a sequence $p_n\downarrow p$. 

        In the first case,  by Lemma \ref{lem:1}, $g(p) = p > q \geq g(q)$
        so that $p$ is a PSIL of $g$.
        In the second case, for any $q>p$ there exists some $p_n \in (p,q)$ so that $g(p) \leq p$ and $g(q) \geq p_n$, so that $p$ is a PSIR of $g$. 

        Next, we show $\mathcal I(g)\subseteq K$. Suppose $p\notin K$ with $p\in (0,1)$. Since $K$ is closed, there exists $q,r$ with $0<q<p<r<1$ such that $(q,r)\cap K = \emptyset$. Then $g$ is constant on $(q,r)$ and $p\notin \mathcal I(g)$. 
        \item Apply (a) to $\mathcal I(g)$, we have
        $$\mathcal I( \mathcal V( \mathcal I (g))) = \mathcal I(g).$$
        By Remark \ref{rmk:3}, we have $\mathcal V(I(g)) \cong g$. 
        \qedhere
    \end{enumerate}
\end{proof}

We need one more technical lemma before stating and proving the equivalence of set-indexed and function indexed partial comonotonicity. 

\begin{lemma}\label{lem:Z}
    Let $K \in \mathcal S$ and for each $q\in K$, $A_q\subseteq \Omega$ has probability $(1-q)$ with $p,q\in K, p<q$ implies $A_p \supseteq A_q$. The random variable $Z$ given by 
    $$Z(\omega) = \sup \{ q\in K: \omega \in A_q\}.$$
    has the following properties.
    \begin{enumerate}[(a)]
        \item For $q\in K$, we have
        $$\P(Z<q) \leq q \leq \P(Z\leq q).$$
        \item If  for some $q\in K$, $(q-\epsilon, q)\cap K = \emptyset$ for some $\epsilon > 0$, then we have $$\P(Z<q)= q.$$ 
        \item If $(p-\epsilon, p) \cap K = \emptyset$ for some $\epsilon  > 0$, let $r = \sup((0,p)\cap K)$, we have
        $$\P(Z \leq r) \geq p\mbox{~and~} \P(Z < r) < p.$$
    \end{enumerate}
\begin{proof}
    \begin{enumerate}[(a)]
        \item \label{a} If $\omega \in A_q$, then $Z(\omega) \geq q$, so that $\{Z < q\} \subseteq A_q^c$, proving the first inequality.
        If $\omega \notin A_q$, then $Z(w) \leq q$, so that $A_q^c \subseteq \{Z \leq q\}$, proving the second inequality.
        \item \label{b} If $\omega \notin A_q$, then $Z(w)<q$ since there is no increasing sequence in $K$ converging to $q$. Combining with part \eqref{a} we have 
        $\P(Z<q) = q$. 
        \item \label{c}
        Let $q = \inf ((p,1)\cap K)$. The set $K$ has empty intersection with $(r,q)$ and $Z$ cannot take any values in $(r,q)$, so that by \eqref{a} and \eqref{b}, we have
        $$\P(Z\leq r) = \P(Z<q) = q \geq p$$
        and 
        $$\P(Z<r) \leq r < p.$$
    \end{enumerate}
\end{proof}

\end{lemma}

\begin{theorem} \label{thm:3} For $K\in \mathcal S$ and $g \in \mathcal G$, we have the following equivalences:
    \begin{enumerate}[(i)]
        \item \label{item:KtoVK}$K$-concentration is equivalent to  $\mathcal V(K)$-comonotonicity;
        \item \label{item:gtoIg}$g$-comonotonicity is equivalent to $\mathcal I(g)$-concentration.
    \end{enumerate}
\end{theorem}

\begin{proof}
(\ref{item:KtoVK}) ($\Rightarrow$): Let $g = \mathcal V(K)$ and suppose $\mathbf X$ is $K$-concentratated. Let $\{A_q\}_{q\in K}$ be a nested collection of common tail events.  Define the random variable $Z$ to be
$$Z(\omega) = \sup \{ q\in K: \omega \in A_q\}.$$
We first show that $Q_Z = g$. Fix $p\in [0,1]$. If $p=0$, $g(0)=0$ and $Q_Z(0) = \VaR_0^+(Z) = 0$. 

Suppose $p\in K \setminus \{0\}$, we wish to show that $\VaR_p(Z) = g(p)$. It suffices to show that the set 
$$S_p = \{x\in \R: \P(Z\leq x) \geq p\}$$ has infimum $g(p)$.

If $p\in K$ and there is a strictly increasing sequence $\{p_n\}$ in $K$ with $p_n \uparrow p$, we have $g(p) = p$. By Lemma \ref{lem:Z} part \eqref{a}, $p\in S_p$. For any $q<p$, there is some $p_n\in (q,p)$. We have
$$\P(Z\leq q) < \P(Z< p_n) \leq p_n < p $$
so that $q\notin S_p$. This shows that $p = \inf S_p$. 

In all other cases,  Lemma \ref{lem:Z} part \eqref{c} shows that $\inf S_p = g(p)$.

Next, we show that each $X_i$ is comonotonic with $Z$. We verify comonotonicity using the formulation given in \cite{S86}, that is, $$(Z(\omega_1)-Z(\omega_2))(X_i(\omega_1)-X_i(\omega_2)) \geq 0 \mbox{~almost~surely.}$$ Fix $i$, and take $\omega_1, \omega_2 \in \Omega$. Without loss of generality, assume $p_1 = Z(\omega_1) < Z(\omega_2)= p_2$. If there exists $q\in (p_1, p_2) \cap K$, we have that $p_1 \in A_q^c$ and $p_2 \in A_q$, so that $X_i(\omega_1) \leq X_i(\omega_2)$. 
Otherwise, $(p_1, p_2)\cap K = \emptyset$. 
In this case, $\omega_1 \notin A_{p_2}$ and $\omega_2 \in A_{p_2}$, we again have  $X_i(\omega_1) \leq X_i(\omega_2)$.

(\ref{item:KtoVK})($\Leftarrow$): Let $\mathcal V(K) = g$ and suppose $\mathbf X=(X_1, \dots, X_d)$ is $g$-comonotonic. Let $Z$ be the random variable comonotonic with each $X_i$ with $Q_Z=g$. For each $q\in K$, define $A_q = \{\omega: Z(\omega) > q\}$. By the construction of $\mathcal V(K)$, $\P(A_q) = 1-q$ and for any $\omega_1 \in A_q, \omega_2 \in A_q^c$, $Z(\omega_1) > Z(\omega_2)$. Comonotonicity  of $(Z, X_i)$ implies that  $X_i(\omega_1) \geq X_i(\omega_2)$ , so that $A_q$ is a $q$-tail event for $X_i$. This holds for all $i=1, \dots, d$ and $q\in K$, so that $\mathbf X$ is $K$-concentrated.

(\ref{item:gtoIg}):  
By Proposition \ref{prop:adj} \eqref{6b}, $g \cong \mathcal V(\mathcal I(g))$. We have 
\begin{align*}
    g\mbox{-com} &\cong \mathcal V(\mathcal I(g))\mbox{-com} &  \text{[Remark \ref{rmk:3}]}\\
    &\cong \mathcal I(g)\mbox{-con} & \text{Part (\ref{item:KtoVK})}
\end{align*}
\end{proof}
\begin{corollary}\label{cor:1}
   Let $K_1, K_2\in \mathcal S$. Then 
   $$K_1 = K_2 \mbox{~if~and~only~if~}K_1\mbox{-con}\cong K_2\mbox{-con}.$$
\end{corollary}
\begin{proof} Let $g_1 = \mathcal V(K_1)$ and $g_2 = \mathcal V(K_2)$. We have 
\begin{align*}
    K_1 = K_2 &\iff \mathcal I(\mathcal V(K_1))= \mathcal I(\mathcal V(K_2)) & \text{Proposition \ref{prop:adj} \eqref{6a}}\\
    &\iff \mathcal I(g_1) = \mathcal I(g_2)\\
    &\iff g_1 \cong g_2 & \text{Remark \ref{rmk:3}}\\
    &\iff g_1\mbox{-com}\cong g_2\mbox{-com} & \text{Remark \ref{rmk:3}}\\
    &\iff K_1\mbox{-con}\cong K_2\mbox{-con}, & \text{Theorem \ref{thm:3}}
\end{align*} 
proving the desired equivalence.
\end{proof}

\begin{remark}
    Let $K\in \mathcal S$. Using the terminology of copulas, $K$-concentration represents a collection of copulas $\mathcal C_K$, in the sense that a random vector $\mathbf X$ is $K$-concentrated if and only if the copula of $\mathbf X$ is a member of $\mathcal C_K$. The best tool to describe $\mathcal C_K$ is ordinal sums from \citet[Section 3.2.2]{N06}. Let the open intervals $\{B_i:=(a_i,b_i)\}$ be a disjoint partition of $K^c$ and closed intervals $\mathcal K=\{K_i := (c_i, d_i)\}$ be a disjoint partition of $K$. Replace each $B_i$ by its closure $[a_i,b_i]$. 
    Let $\{J_i\}$ be the collection of non-degenerate closed intervals among all $\overline{B_i}$ and $K_i$. Any collection of copulas $\{C_i\}$ with the same index of $\{J_i\}$ defines a copula through the ordinal sum. Such an ordinal sum copula is a member of $\mathcal C_K$ if the following holds:
    \begin{align*}
        &\mbox{each $C_i$ is the comonotonic copula when   $J_i \in \mathcal K$.} 
    \end{align*}
Conversely, all members of $\mathcal C_K$ is of this form. Figure \ref{fig:ordinal} illustrates this idea. Take $K = [0,1/3]\cup [2/3,1]$. Put any copula inside the middle box, the resulting ordinal sum copula is in $\mathcal C_K$.
\begin{figure}
    \centering
        \begin{center}
    \begin{tikzpicture}
\draw[<->] (0,3.4) -- (0,0) -- (3.4,0);
\draw[gray,dotted] (0,0.6) -- (3,0.6);
\draw[gray,dotted] (0,1.2) -- (3,1.2);
\draw[gray,dotted] (0,1.8) -- (3,1.8);
\draw[gray,dotted] (0,2.4) -- (3,2.4);
\draw[gray,dotted] (0.6,0) -- (0.6,3);
\draw[gray,dotted] (1.2,0) -- (1.2,3);
\draw[gray,dotted] (1.8,0) -- (1.8,3);
\draw[gray,dotted] (2.4,0) -- (2.4,3);
\draw[gray,dotted] (3,0) -- (3,3);
\draw[gray,dotted] (0,3) -- (3,3); 

\node[below] at (3,0) {$1$}; 
\node[below] at (0,0) {$0$}; 
\node[red, below] at (0.5,0) {$K$};
\node[red, below] at (1.5,0) {$K^c$};
\node[red, below] at (2.5,0) {$K$};
\draw[red, dotted] (1,0)--(1,1);
\draw[red, dotted] (2,0)--(2,1);

\draw[blue] (0,0) -- (1,1); 
\draw[blue] (2,2) -- (3,3);

\draw[blue] (1,1) -- (2,1) -- (2,2) -- (1,2) -- cycle;

    \end{tikzpicture}
\end{center}  
    \caption{$K$-concentration as a collection of Ordinal sums}
    \label{fig:ordinal}
\end{figure}
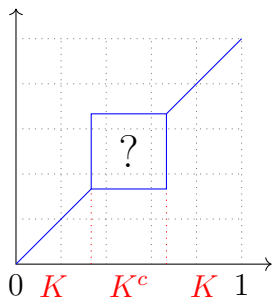

\end{remark}

\section{On spectral risk measures}\label{sec:spectral}

So far, our study has not imposed monotonicity on the risk functionals $I_h$ for $h\in \mathcal H^{\rm bv}$.  Next, we consider a special class of $I_h$ that is coherent, \footnote{We also assume that $h$ is continuous at $0$. For $I_h$ coherent, $h$ is an increasing concave function  with $h(0)=0$ and $h(1)=1$, the only possible discontinuity is at $0$. A jump corresponds to adding a multiple of the essential supremum.}  that is, the class of  spectral risk measures. We obtain an equivalent condition on the random vector for which a given spectral risk measure is additive.

Spectral risk measures, as introduced in \cite{A02}, take the form 
$$\rho_g(X) = \int_0^1  g(t)Q_X(t) \d t$$ for some increasing function $g:[0,1] \to \mathbb R^+$ that integrates to $1$.  Since there are only countably many discontinuities, we may modify the risk spectrum $g$ to be left-continuous without changing the value of the integral, so that we may assume that $g\in \mathcal G$. Denote the collection of risk spectrums by $\mathcal G_1$, which can be considered as a subset of $\mathcal G$.  

Spectral risk measures is   a special case of a distortion riskmetric. Define a distortion riskmetric $I_h$ by letting the distortion function $h$ have   its conjugate function $\hat h$ to be the integral of $g$, we have that $\rho_g = I_h$. In the special setting of spectral risk measures, we will be able to provide a result stronger than Theorem \ref{thm:KCA}. We first rephrase Theorem \ref{thm:KCA} in this context.

\begin{corollary} Let  $g_1\in \mathcal G_1$ and  $g_2\in \mathcal G$.  
    The spectral risk measure $\rho_{g_1}$ is $g_2$-additive if and only if $g_1 \precsim g_2$.
\end{corollary}
\begin{proof}
    Let $I_h$ be the Choquet integral such that $I_h = \rho_g$. That is, $\hat h^\prime(t) = g(t)$. By Theorem \ref{thm:KCA}, for $K \in \mathcal S$, $\rho_g$ is $K$-additive if and only if $\hat h$ is linear on each connected component of $K^c$, which is equivalent to $g$ being constant on each connected component of $K^c$. But $g$ is constant on each connected component of $\mathcal I(g)^c$, therefore $\rho_g$ is $K$-additive if and only if $K^c \subseteq \mathcal I(g)^c$. Hence $\rho_{g_1}$ is $g_2$-additive if and only if $\mathcal I(g_1) \subseteq \mathcal I(g_2)$, if and only if  $g_1 \precsim g_2$.
\end{proof}

\begin{theorem} \label{EgCA} Let $g\in \mathcal G_1$ and $\mathbf X$ be a random vector. 
    The spectral risk measure $\rho_g$ is additive on $\mathbf X$ if and only if $\mathbf X$ is $g$-comonotonic.
\end{theorem}
The proof uses the following lemma.
\begin{lemma}\label{728}(\cite{MFE15}, Theorem 7.28) 
     Let $X,Y$ be random variables with finite positive variance, and $X^c, Y^c$ be comonotonic copies of $X,Y$.  Then $\E (XY) \leq \E(X^cY^c)$ with equality holding if and only if $X,Y$ are comonotonic.
\end{lemma}

\begin{proof}[Proof of Theorem \ref{EgCA}]
Let $Z$ be some random variable with $Q_Z(t) = g(t)$. Given a random vector $\mathbf X = (X,Y)$, consider the following expressions:
\begin{align*}
    \rho_g(X+Y) &= \int_0^1 Q_Z(t)Q_{X+Y} (t) \d t\\
    &\geq \E\big(Z(X+Y)\big) = \E(ZX)+\E(ZY),
\end{align*}
and
\begin{align*}
    \E(ZX)+\E(ZY) &\leq \int_0^1 Q_Z(t)Q_{X} (t) \d t+\int_0^1 Q_Z(t)Q_{Y} (t) \d t\\
    &= \rho(X)+\rho(Y).
\end{align*}
$(\Leftarrow):$ If $(X,Y)$ is $g$-comonotonic, pick $Z$ as in the definition so that $(X,Z)$ and $(Y,Z)$ are comonotonic. Hence, $Z$ and $X+Y$ are comonotonic, so all inequalities become equalities.

$(\Rightarrow):$ Conversely, let $U$ be the distributional transform (\citet[Lemma A.32]{FS16}) of $X+Y$, i.e., $X+Y = Q_{X+Y}(U)$. Let $Z = g(U)$, $Z$ and $X+Y$ are comonotonic because they are both increasing functions of $U$. Then the first inequality is an equality, but $Z$ cannot be comonotonic with both $X$ and $Y$, so the second inequality is strict.
\end{proof}

\begin{remark}
    Theorem \ref{EgCA} is stronger in the following sense.  In Theorem \ref{thm:KCA}, an $\mathbf X$ not satisfying the corresponding $K$-concentration may still be additive under $I_h$: it only requires that any $\mathbf X$ satisfying $K$-concentration be additive, and that some $\mathbf X$ not satisfying $K$-concentration be not additive. However, in the context of spectral risk measures, any $\mathbf X$ not satisfying $K$-concentration cannot be additive.
\end{remark}

\begin{example}
    For $p \in (0,1)$, let  $g_p(x) = \cdot \mathds{1}_{[p,1]}(x)/(1-p)$  for  $x\in [0,1]$. The spectral risk measure induced by $g_p$ is $\ES_p$. The corresponding closed set $K$ is $\mathcal I(g_p) = \{p\}$.  Theorem \ref{thm:KCA} would imply that $\ES_p$ is $p$-additive, and that any $p$-additive distortion riskmetric $\rho$ is of the form     $\rho = \alpha \ES_p + \beta \E$ for some $\alpha,  \beta \in \R$.  Theorem \ref{EgCA} would further imply that, if $\mathbf X$ is not $p$-concentrated, then $\ES_p$ is not additive on $\mathbf X$.
\end{example}

\section{Applications}\label{sec:app}

\subsection{Additivity of well-known functionals}
 Recall that comonotonicity is equivalent to $[0,1]$-concentration, so that all distortion riskmetrics are $[0,1]$-additive. As an application of Theorem \ref{thm:KCA}, we discuss several well-known distortion riskmetrics listed in \cite{WWW20b}, and show that some of them are in fact $K$-additive for $K$ smaller than $[0,1]$. 
\begin{example}
    Fix some $p\in (0,1)$ and take $\rho = \VaR_p$. We proved directly in Theorem \ref{thm:quant} that $\rho$ is $K$-additive if and only if there exists a sequence $\{p_n: n\in \mathbb N\} \subseteq K$ such that $p_n \uparrow p$. We illustrate how we can arrive at the same conclusion using Theorem \ref{thm:KCA}. 

    The distortion function corresponding to $\rho$ is 
    $h: t \mapsto \mathds{1}_{\{t>1-p\}}.$ We have that $h$ is linear on $[0,1-p]$ and $(1-p,1]$. Let $K = \{p\}$. Why is $\rho$ not $K$-additive? Ignoring the points $\{0,1\}$, the open interval decomposition of $K^c$ is $B_1 = (0,p)$ and $B_2 = (p,1)$. So that $$\mathsf{cl}(1-B_1) = [1-p, 1],~~\mathsf{cl}(1-B_2) = [0, 1-p].$$
    However, $h$ is linear on $(1-p,1]$ but not on $[1-p,1]$. According to Theorem \ref{thm:KCA}, for $\rho$ to be $K'$-additive, a necessary condition is that for any $\epsilon > 0$, $[1-p, 1-p+\epsilon]$ is not one of the $\mathsf{cl}(1-B_i)$'s for $K'$, since $h$ is not linear on $[1-p, 1-p+\epsilon]$. Equivalently, $(p-\epsilon, p) \cap K' \neq \emptyset$ for all $\epsilon > 0$. This implies that there exists $\{p_n : n \in \mathbb N\} \subseteq K$ such that $p_n \uparrow p$. 
\end{example}

\begin{example}
    Take $\rho$ to be the mean-median deviation. That is, 
    $$\rho(X) = \min_{x\in \R} \E[|X-x|].$$
    The distortion function of $\rho$ is $h: t\mapsto t \wedge (1-t)$, so that $h$ is linear on $[0,0.5]$ and $[0.5,1]$. Therefore $\rho$ is $K$-additive for $K$ being the singleton set $\{0.5\}$. 
\end{example}

\begin{example}
    Let $\rho$ be the Gini Shortfall, defined by $$\mathrm{GS}_\alpha^\lambda (X)= \mathrm{ES}_\alpha(X) + \lambda \mathbb{E}\Bigl[ \bigl| X_\alpha^* - X_\alpha^{**} \bigr| \Bigr], \quad \alpha \in (0,1),\ \lambda \geq 0,$$
    where $X_\alpha^*$ and $X_\alpha^{**}$ are iid copies of $Q_X(U_\alpha)$, with $\alpha$ uniform on $[\alpha, 1]$. 
    The distortion function of $\rho$ is
    $$h(t) = \frac{t}{1 - \alpha}\wedge 1 + \frac{2\lambda t (1 - t - \alpha)_{+}}{(1 - \alpha)^2}.$$
    For $\lambda \neq 0$, $h$ is linear only on $[1-\alpha, 1]$, so that the smallest $K\in  \mathcal S$ for which $\rho$ is $K$-additive is $K = [\alpha, 1]$. This was upper comonotonicity in Example \ref{ex:1}.
\end{example}
\begin{example}
    Take $\rho$ to be the proportional hazard principle, also called $\mathrm{MAXVAR}$,  defined by 
    $$X \mapsto \frac{1}{\alpha} \int_{0}^{1} (1 - t)^{(1-\alpha)/\alpha} F_X^{-1}(t) \, \d t, \quad \alpha \geq 1.$$
    The distortion function of $\rho$ is $h: t \mapsto t^{1/\alpha}$. We see that if $\alpha \neq 1$, $h$ is not linear on any nondegenerate interval and for $\rho$ to be $K$-additive, $K$ has to be $[0,1]$ (recall that $K\in \mathcal S$ is closed). Hence $\rho$ is comonotonic additive but not $K'$-additive for any closed proper subset $K'\subsetneq [0,1]$.
\end{example}

\subsection{Selection of spectral risk measures via $K$-additivity}
A practitioner in risk management is choosing a risk measure. 
Law-invariance and coherence are desirable properties of risk measures, we suppose that the practitioner limits her choices  within the class of law-invariant coherent risk measures.
Suppose further that for some $K\in \mathcal S$, $K$-additivity is also a desirable property. Since comonotonicity implies $K$-concentration, $K$-additivity implies comonotonic additivity. This naturally restricts the choices to the class of spectral risk measures. 

Two questions arise. First, how is the set $K$ determined to suit the practitioner's needs? Second, once $K$ is determined, how does this help the practitioner choose a desirable risk measure. 

The first question is broad and we discuss one consideration.  Coherent risk measures are subadditive, and when additivity occurs, it is interpreted that the involved random losses do not diversify risks. Therefore, $K$ can be chosen to model the collection of undesirable dependencies. 

The second question is answered by Theorem \ref{EgCA} and Theorem \ref{thm:3}. That is, all risk measures that are law-invariant, coherent and $K$-additive are given by 
$$\{\rho_g: g\in \mathcal G_1 \text{~and~} \mathcal I(g) \subseteq K\}.$$

If we are precise about the undesirable dependencies $K$, in the sense that we want our risk measure $\rho$ to not be $K'$-additive for any of $\{K'\in \mathcal S:~  K'\subsetneq K\}$, we may restrict the collection of risk measures to be
$$M_K := \{\rho_g: g\in \mathcal G_1 \text{~and~} \mathcal I(g) = K\}.$$

When $K$ is a finite set, $M_K$ has a simple description. 
\begin{theorem}
    Suppose $K = \{\alpha_0, \alpha_1, \dots, \alpha_n, \alpha_{n+1}\}$. Then $M_K$ is of the form 
    $$M_K = \left \{\lambda_0\E  + \lambda_1 \ES_{\alpha_1} + \cdots \lambda_n \ES_{\alpha_n} : \sum_{i=0}^n \lambda_i = 1,~ \lambda_0 \geq 0,  \lambda_1, \dots , \lambda_n > 0\right \}.$$
\end{theorem}
\begin{proof}
    Without loss of generality, suppose $0= \alpha_0 <\alpha_1 < \cdots < \alpha_n< \alpha_{n+1}=1$. Since $\mathcal I(g) = K$ and $g\in \mathcal G_1 \subseteq \mathcal G$, the function $g$ is of the form 
    $$g = \gamma_0 \mathds{1}\{[0, \alpha_1]\}+ \left(\sum_{i=1}^{n} \gamma_i \mathds{1}\{(\alpha_i, \alpha_{i+1}] \}\right) $$
    for some $0 \leq \gamma_0 < \gamma_1 < \cdots < \gamma_n$. Let $\lambda_0 = \gamma_0$ and 
    $$ \lambda_i = (\gamma_i - \gamma_{i-1})(1-\alpha_i) \mbox{~for~}i=1, \dots, n.$$   
    We can rewrite $g$ as
    $$g = \lambda_0 \mathds 1\{[0, 1]\} + \left(\sum_{i=1}^n \lambda_i \frac{\mathds{1}\{(\alpha_i, 1]\}}{1-\alpha_i}\right)$$
    so that $\rho_g = \lambda_0\E  + \lambda_1 \ES_{\alpha_1} + \cdots \lambda_n \ES_{\alpha_n}$.
    Clearly $\lambda_0 \geq 0$ and $\lambda_i > 0$ for $i = 1, \dots, n$. Finally, we check the sum: 
    \begin{align*}
        \sum_{i=0}^n \lambda_i &= \gamma_0 + \sum_{i=1}^n \lambda_i \frac{1-\alpha_i}{1-\alpha_i}\\
        &= \int_0^1 g(t) \d t = 1. \qedhere
    \end{align*}
\end{proof}

The following example illustrates these ideas.
\begin{example}
    Let $\{A,B,C,D\}$ be a partition of $\Omega$ with $\P(A)=\P(B)=\P(C)=0.05$. Suppose we hold the random loss $X$ in our portfolio and wish to diversify our risk with candidate random losses $X_1, X_2, X_3$ and $X_4$,  defined in the following table, where $U_A, U_B, U_C$ are random variables with conditional distributions $U_A|A, U_B|B, U_C|C$ all being uniform $[0,1]$. 
\begin{table}[htbp]
\centering
\begin{tabular}{ccccc}
\toprule
& D & C & B & A\\
\midrule
$X$ & 0 & $0.5+U_C$ & $1.5+U_B$ & $2.5+U_A$\\
$X_1$ & 0 & 1 & 2 & 3\\
$X_2$ & 0 & 1 & 3 & 2\\
$X_3$ & 0 & 2 & 1 & 3\\
\bottomrule
\end{tabular}~~
\end{table}

Suppose that we are concerned about the catastrophic events $A$ and $A\cup B$, so that $X_1$ should be the most undesirable asset for the purpose of diversification.  Naturally, the relevant dependence structure is $K$-concentration with $K = \{0.9, 0.95\}$. The corresponding $g = \mathcal V(K)$ is $0.9 \cdot \mathds 1\{(0.9, 0.95]\} + 0.95 \cdot \mathds 1\{(0.95,1]\}$. Thus, the collection of risk spectrums equivalent to $g$ are of the form 
$$g = \gamma_0 \cdot \mathds 1\{[0,0.9]\} + \gamma_1 \cdot \mathds 1\{(0.9, 0.95]\} + \gamma_2 \cdot \mathds 1\{(0.95,1]\} $$ with $0 \leq \gamma_0 < \gamma_1 < \gamma_2$, $a+b+18c = 20$. We can rewrite $g$ as 
\begin{align*}
    g &= \gamma_0 \mathds{1}\{[0,1]\}+ (\gamma_1-\gamma_0) \cdot \mathds 1\{(0.9, 1]\} + (\gamma_2-\gamma_1) \cdot \mathds 1\{(0.95,1]\}
\end{align*}
Take $\lambda_0 = \gamma_0, \lambda_1 = (\gamma_1-\gamma_0)/10$ and $\lambda_2 = (\gamma_2-\gamma_1)/20$, we see that the corresponding spectral risk measure is 
$$\rho_g = \lambda_0 \E + \lambda_1 \ES_{0.9}+ \lambda_2 \ES_{0.95}.$$ Let $\rho$ be the risk measure for  $\lambda_0 = 0, \lambda_1 = \lambda_2 = 1/2$. The risk values are computed in the following table.
\begin{table}[htbp]
    \centering
    \begin{tabular}{cccc}
\toprule
& $\ES_{0.9}$ &$\ES_{0.95}$  & $\rho$ \\
\midrule
$X_1+X$ & 5 & 6 & 5.5 \\
$X_2+X$ & 5 & 5 & 5 \\
$X_3+X$ & 4.5 & 6 & 5.25 \\
\bottomrule
$2\rho(X)$ & 5 & 6 & 5.5
\end{tabular}
\end{table}

We see that $\ES_{0.9}$ is additive on $(X, X_1)$ and $(X, X_2)$, identifying $X_1$ and $X_2$ as undesirable,  and $\ES_{0.95}$ is additive on $(X, X_1)$ and $(X, X_3)$, identifying $X_1$ and $X_3$ as undesirable. On the other hand, $\rho$ uniquely identifies $X_1$ as the undesirable asset for diversification.
\end{example}

\section{Extension to a larger space of random variables}\label{sec:ext}
For $h\in \mathcal H^{\rm bv}$ and $X \in L^\infty$, the expression $I_h(X)$ as defined in \eqref{eq:DRM} is always finite. However, in practice, many commonly used distributions are not essentially bounded. In this section, we extend the domain of the distortion riskmetrics to spaces larger than $\mathcal X = L^\infty$. We first allow risk functionals to take the value $\infty$.

Define $C_h: L^0 \to \R  \cup \{-\infty, \infty, {\rm undefined}\}$ by the same formula $$
C_h(X) = \int_{-\infty}^0\left [(h(\P(X>x))-h(1)\right ]\d x + \int_0^\infty h(\P(X>x))\d x
$$
where $-\infty + \infty$ is undefined.
Given $h\in \mathcal H^{\rm bv}$, we choose a convex cone $\mathcal X_h$ with $L^\infty \subseteq \mathcal X_h \subseteq L^0$ such that for any $X\in \mathcal X_h$, $C_h(X)\in (-\infty, \infty]$. The convex cone requirement ensures that if $X,Y\in \mathcal X_h$, then $X+Y$ is also in $\mathcal X_h$. We define a distortion riskmetric associated to $h$ by $$J_h: \mathcal X_h \to (-\infty, \infty],~ X \mapsto C_h(X).$$ 

We show that it is possible to take  $\mathcal X_h$ to include a larger space than $L^\infty$. Let $\delta_1 \in \mathcal{H}^{\rm bv}$ be defined by $$\delta_1(t) = \begin{cases}
    1 & {\rm if~} t=1\\
    0 & {\rm otherwise},
\end{cases}$$ the signed Choquet integral $C_{\delta_1}$ is the essential infimum and we can take $\mathcal X_{\delta_1}$ to be the set of all random variables $X$ such that ${\rm ess\text{-}inf}(X)>-\infty.$  Let 
$$c = \sup_{t\in [0,1]} |h(t)|.$$ Since $h$ has finite variation, $c  < \infty$. For $X\in \mathcal X_{\delta_1}$, we have  
\begin{align*}
    \int_{-\infty}^0\left [(h(\P(X>x))-h(1)\right ]\d x &= \int_{{\rm ess\text{-}inf} (X)\wedge 0}^0\left [(h(\P(X>x))-h(1)\right ]\d x\\
    &\geq 2c\left ( {\rm ess\text{-}inf}(X)\wedge 0\right)>-\infty
\end{align*}
and 
$$\int_0^\infty h(\P(X>x))\d x \geq 0$$
so that $C_h(X) \in (-\infty, \infty]$. Therefore $\mathcal X_h$ can be chosen to satisfy $\mathcal X_{\delta_1} \subseteq \mathcal X_h \subseteq L^0$. 
In particular, commonly used distributions such as Gamma and Pareto are non-negative, and are included in $\mathcal X_h$.


The definition of $K$-concentration and $g$-comonotonicity remains valid when the random vector $\mathbf X$ is taken to be in $(\L^0)^d$. The definitions of additivity of a distortion riskmetric remains unchanged, noting that $\mathcal X$ is the respective $\mathcal X_h$ instead of $L^\infty$. 
It is straightforward to verify that the $(\Leftarrow)$ direction of Theorem \ref{thm:KCA} remains valid. The $(\Rightarrow)$ direction is illustrated by the following theorem.
\begin{theorem}
    Fix $K \in \mathcal S$, suppose there exists a subset $\mathcal H_K$ of $\mathcal H^{\rm bv}$ such that $I_h$ being $K$-additive implies that $h\in \mathcal H_K$. Then $J_h$ being $K$-additive also implies that $h\in \mathcal H_K$.
\end{theorem}
\begin{proof}
    Let $J_h$ be $K$-additive and $(X,Y)$ be any $K$-concentrated random vector in $L^\infty$. Since $L^\infty \subseteq \mathcal X_h$ and $J_h$ is $K$-additive, $J_h$ is additive for $(X,Y)$. This shows that the restriction $J_h|_{L^\infty} = I_h$ is $K$-additive, and hence $h\in \mathcal H_K$.
\end{proof}

However, since Lemma \ref{728} requires the random variables $X,Y$ to have finite variance, Theorem \ref{EgCA} does not generalize to this setting. Let $X,Y$ be non-negative random variables with infinite mean. Since any spectral risk measure $\rho$ satisfies $\rho \geq \E$, we have
$$\rho(X+Y) = \infty = \infty + \infty = \rho(X)+\rho(Y)$$ regardless of the dependence structure between $X$ and $Y$.

\section{Conclusion}\label{sec:conclusion}

We provided two equivalent formulations of partial comonotonicity: $K$-concentration, indexed by closed sets $K\subseteq [0,1]$ and $g$-comonotonicity, indexed by increasing left-continuous functions $g:[0,1]\to \R$. For a dependence structure $D$ from the above classes, we characterized the class of distortion riskmetrics that are additive for random vectors satisfying $D$.

In place of additivity, one may consider another interesting property: aversion, where risks are highest under a certain dependence structure $D$. Formally, a risk functional $\rho$ is $D$-averse if for any $(X,Y)$ satisfying $D$, any $X^\prime \laweq X$ and $Y^\prime \laweq Y$ satisfies
$$\rho(X+Y)\geq \rho(X^\prime+Y^\prime).$$
Under the assumption of subadditivity for the risk functional, $D$-additivity implies $D$-aversion. The validity of the converse condition reduces to whether additivity is always achievable, which is an interesting problem in itself.

When $D$ is comonotonicity, $D$-aversion was studied in \cite{MW20} and termed diversification consistency. When $D$ is $p$-concentration, $D$-aversion was studied in \cite{HWWW24}, and characterization results were established. For generic partial comonotonicities, the analysis seems more challenging and a theory remains to be established.

\section*{Acknowledgements}

The author  acknowledges financial support from the Queen Elizabeth II Graduate Scholarship in Science and Technology (QEII-GSST), the University of Waterloo President's Graduate Scholarship, the Statistics and Actuarial Science Chair's Award, the Mathematics Senate Graduate Scholarship, and Research Assistantship funding provided by the author’s supervisor Ruodu Wang. \\
The author would like to thank his supervisor for valuable guidance and support throughout the preparation of this paper.


\begin{thebibliography}{} 


\bibitem[\protect\citeauthoryear{Acerbi}{2002}]{A02}
Acerbi, C. (2002). Spectral measures of risk: A coherent representation of subjective risk aversion. \textit{Journal of Banking and Finance}, \textbf{26}(7), 1505--1518.



\bibitem[\protect\citeauthoryear{Amarante and Liebrich}{Amarante and Liebrich}{2024}]{AL24}
Amarante, M. and Liebrich, F. B. (2024). Distortion risk measures: Prudence, coherence, and the expected shortfall. \emph{Mathematical Finance}, \textbf{34}(4), 1291--1327.


\bibitem[\protect\citeauthoryear{Artzner et al.}{Artzner et al.}{1999}]{ADEH99}
{Artzner, P., Delbaen, F., Eber, J.-M. and Heath, D.} (1999). Coherent measures of risk. \emph{Mathematical Finance}, \textbf{9}(3), 203--228.


\bibitem[\protect\citeauthoryear{Bellini et al.}{2021}]{BKMS21} 
Bellini, F., Koch-Medina, P., Munari, C. and Svindland, G. (2021). Law-invariant riskmetrics that collapse to the mean. \emph{Insurance: Mathematics and Economics}, \textbf{98}, 83--91.

\bibitem[\protect\citeauthoryear{Cheung}{2009}]{C09}
Cheung, K. C. (2009). Upper comonotonicity. \emph{Insurance: Mathematics and Economics}, \textbf{45}(1), 35--40.



\bibitem[\protect\citeauthoryear{Dhaene et al.}{Dhaene et al.}{2002}]{DDGKV02}
{Dhaene, J., Denuit, M., Goovaerts, M.~J., Kaas, R. and Vynche, D.} (2002). {The concept of comonotonicity in actuarial science and finance: Theory}. {\em Insurance: Mathematics and Economics},
\textbf{31}(1), 3--33.

\bibitem[\protect\citeauthoryear{Dhaene et al.}{Dhaene et al.}{2012}]{DKLT12}
{Dhaene, J. and Kukush, A., Linders, D. and Tang, Q.} (2012). Remarks on quantiles and distortion risk measures. \emph{European Actuarial Journal}, \textbf{2}(2), 319--328.

\bibitem[\protect\citeauthoryear{Dhaene et al.}{Dhaene et al.}{2006}]{DVGKTV06}
{Dhaene, J., Vanduffel, S., Goovaerts, M.J., Kaas, R., Tang, Q. and Vynche, D.} (2006). {Risk measures and comonotonicity: A review}. \emph{Stochastic Models}, \textbf{22}, 573--606.





\bibitem[\protect\citeauthoryear{Esary et al.}{Esary et al.}{1967}]{EPW67}
Esary, J. D., Proschan, F. and Walkup, D. W. (1967). Association of random variables, with applications. \emph{Annals of Mathematical Statistics}, \textbf{38}(5), 1466--1474.


 \bibitem[\protect\citeauthoryear{F\"ollmer and Schied}{F\"ollmer and Schied}{2016}]{FS16} F\"ollmer, H.~and Schied, A.~(2016). \emph{Stochastic Finance: An Introduction in Discrete Time}. Fourth Edition.  {Walter de Gruyter, Berlin}.

 

\bibitem[\protect\citeauthoryear{Han et al.}{Han et al.}{2024}]{HWWW24}
{Han, X., Wang, B., Wang, R. and Wu, Q.} (2024). Risk concentration and the mean‐expected shortfall criterion. \emph{Mathematical Finance}, \textbf{34}(3), 819-846.



\bibitem[\protect\citeauthoryear{Kusuoka}{Kusuoka}{2001}]{K01}
{Kusuoka, S.} (2001). On law invariant coherent risk measures. \emph{Advances in Mathematical Economics}, \textbf{3}, 83--95.






\bibitem[\protect\citeauthoryear{McNeil et al.}{McNeil et al.}{2015}]{MFE15}
{McNeil, A. J., Frey, R., and Embrechts, P. } (2015). 
(2015). \emph{Quantitative risk management: concepts, techniques and tools}. Revised Edition. Princeton University Press.



\bibitem[\protect\citeauthoryear{Mu et al.}{2024}]{MPST24}
Mu, X., Pomatto, L., Strack, P., and Tamuz, O. (2024). Monotone additive statistics. \emph{Econometrica},  \textbf{92}(4), 995--1031. 


\bibitem[\protect\citeauthoryear{Principi et al.}{2025}]{PWW25}
Principi, G., Wakker, P. P. and Wang, R. (2025). Antimonotonicity for preference axioms: The natural counterpart to comonotonicity. \emph{Theoretical Economics}, \textbf{20}(3), 831-855.


\bibitem[\protect\citeauthoryear{Lauzier et al.}{Lauzier et al.}{2023}]{LLW23}
Lauzier, J.-G., Lin, L. and Wang, R. (2023). Pairwise counter-monotonicity. \emph{Insurance: Mathematics and Economics}, \textbf{111}, 279--287.

\bibitem[\protect\citeauthoryear{Mao and Wang}{Mao and Wang}{2020}]{MW20}
{Mao, T. and Wang, R. } (2020). Risk aversion in regulatory capital principles. \emph{SIAM Journal on Financial Mathematics}, \textbf{11}(1), 169-200.



\bibitem[\protect\citeauthoryear{Nelson}{Nelson}{2006}]{N06} 
Nelsen, R. B. (2006). \emph{An introduction to copulas}. Springer New York.






  
\bibitem[\protect\citeauthoryear{Schmeidler}{Schmeidler}{1986}]{S86}
{Schmeidler, D.} (1986). Integral representation without additivity.
\emph{Proceedings of the
American Mathematical Society}, \textbf{97}(2), 255--261.



\bibitem[\protect\citeauthoryear{Schmeidler}{Schmeidler}{1989}]{S89}
{Schmeidler, D.} (1989). Subjective probability and expected utility without additivity. \emph{Econometrica: Journal of the Econometric Society}, \textbf{57}(3), 571-587.


\bibitem[\protect\citeauthoryear{Wang, Young and Panjer}{Wang et al.}{1997}]{WYP97}
 
Wang, S. S., Young, V. R., and Panjer, H. H. (1997). Axiomatic characterization of insurance prices. \emph{Insurance: Mathematics and Economics}, \textbf{21}(2), 173-183.

\bibitem[\protect\citeauthoryear{Wang and Zitikis}{2020}]{WZ20}
 Wang, R., and Zitikis, R. (2020). Weak comonotonicity. \emph{European Journal of Operational Research}, \textbf{282}, 386--397.

\bibitem[\protect\citeauthoryear{Wang and Zitikis}{Wang and Zitikis}{2021}]{WZ21}
Wang, R., and Zitikis, R. (2021). An axiomatic foundation for the Expected Shortfall. \emph{Management Science}, \textbf{67}(3), 1413-1429.


\bibitem[\protect\citeauthoryear{Wang et al.}{Wang et al.}{2020b}]{WWW20}
Wang, R., Wei, Y., and Willmot, G. E. (2020b). Characterization, robustness, and aggregation of signed Choquet integrals. \emph{Mathematics of Operations Research}, \textbf{45}(3), 993-1015.

\bibitem[\protect\citeauthoryear{Wang et al.}{Wang et al.}{2020a}]{WWW20b}
Wang, Q., Wang, R. and Wei, Y.  (2020a). Distortion riskmetrics on general spaces. \emph{ASTIN Bulletin}, \textbf{50}(4), 827--851. 



\bibitem[\protect\citeauthoryear{Yaari}{Yaari}{1989}]{Y87}
Yaari, M. E. (1987). The dual theory of choice under risk. \emph{Econometrica: Journal of the Econometric Society}, \textbf{55}(1), 95-115.



\bibitem[\protect\citeauthoryear{Zhang and Duan}{Zhang and Duan}{2013}]{ZD13}
Zhang, L. and  Duan, B. (2013). Extensions of the notion of overall comonotonicity to partial comonotonicity. \emph{Insurance: Mathematics and Economics}, \textbf{52}(3), 457-464.

\end{thebibliography}
\end{document}